\documentclass[preprint,12pt]{elsarticle}

\makeatletter
\def\ps@pprintTitle{%
 \let\@oddhead\@empty
 \let\@evenhead\@empty
 \def\@oddfoot{}%
 \let\@evenfoot\@oddfoot}
\makeatother

\usepackage[group-separator={,}]{siunitx}
\usepackage{amsmath,amsfonts,amssymb,amsthm}
\usepackage{graphicx}
\usepackage{url}
\usepackage{mathtools}
\usepackage[textsize=tiny,textwidth=1cm,shadow]{todonotes}
\usepackage{booktabs}
\usepackage{setspace}
\usepackage{geometry}
\usepackage{pgf}
\usepackage{tikz}
\usetikzlibrary{shadows,arrows,decorations,decorations.shapes,backgrounds,shapes,snakes,automata,fit,petri,shapes.multipart,calc,positioning,shapes.geometric}
  \tikzstyle{whitesq}=[rectangle,draw=black,fill=white,thin,inner sep=1.5pt,minimum size=6mm]
  \tikzstyle{whitecirc}=[circle,draw=black,fill=white,thin,inner sep=1.5pt,minimum size=6mm]
  \tikzstyle{fake}=[circle,draw=black,fill=white,thin,inner sep=1.5pt,minimum size=6mm,opacity=0.0]

\newtheorem{theorem}{Theorem}

\newcommand{\pathweight}[1]{w\left[{{#1}}\right]}

\journal{Information Processing Letters}  

\begin{document}

\begin{frontmatter}

\title{Hardness of the Pricing Problem for Chains in Barter Exchanges}

\author[cmu]{Benjamin Plaut}
\author[cmu]{John P. Dickerson}
\author[cmu]{Tuomas Sandholm}

\address[cmu]{Computer Science Department\\Carnegie Mellon University}

\begin{abstract}
Kidney exchange is a barter market where patients trade willing but medically incompatible donors. These trades occur via cycles, where each patient-donor pair both gives and receives a kidney, and via chains, which begin with an altruistic donor who does not require a kidney in return. For logistical reasons, the maximum length of a cycle is typically limited to a small constant, while chains can be much longer.  Given a compatibility graph of patient-donor pairs, altruists, and feasible potential transplants between them, finding even a maximum-cardinality set of vertex-disjoint cycles and chains is NP-hard. There has been much work on developing provably optimal solvers that are efficient in practice. One of the leading techniques has been \emph{branch and price}, where column generation is used to incrementally bring cycles and chains into the optimization model on an as-needed basis.  In particular, only positive-price columns need to be brought into the model.  We prove that finding a positive-price chain is NP-complete. This shows incorrectness of two leading branch-and-price solvers that suggested polynomial-time chain pricing algorithms.
\end{abstract}


\end{frontmatter}


\section{Introduction}%
\label{sec:intro}

In barter markets, participants swap items for items rather than for money.  Barter markets have been organized for many categories of items such as  holiday home time, nurse shifts, used goods, and even shoes.
Kidney exchange is an organized barter market where patients in need of a kidney trade their paired but incompatible donors with other participants in the exchange~\cite{Roth04:Kidney}.  Mathematically, a kidney exchange is typically modeled as a directed \emph{compatibility graph} $G = (V, E)$. The set of vertices $V$ is partitioned into $P$ and $A$, where vertices in $P$ represent patient-donor pairs and vertices in $A$ represent altruistic donors, who enter the exchange without a paired patient. For each $u,v \in P$, the edge $(u,v)$ exists if the donor of pair $u$ is compatible with the patient of pair $v$. Similarly, for each $a\in A$ and $v \in P$, the edge $(a,v)$ exists if altruist $a$ is compatible with the patient of pair $v$. These edges may also have weights, representing the relative value of a potential transplant.

Given a compatibility graph $G$, the \emph{weighted clearing problem} is to find a maximum-weight vertex-disjoint set of cycles of length at most $L$ and chains of length at most $K$.  Even the maximum-cardinality clearing problem with $L \geq 3$ is NP-hard~\cite{Abraham07:Clearing}.  In practice, $L$ is limited to a small constant (typically $L=3$) due to the logistical difficulty of arranging all transplants in a cycle simultaneously; however, due to the non-simultaneous execution of chains, $K$ is often much greater than $L$---and can even be allowed to grow with $|V|$.

Integer programming techniques power optimal clearing engines in fielded kidney exchanges.  A leading approach to solving these integer program models---which can be quite large due to the $O(|P|^L)$ cycles and $O(|A||P|^{K})$ chains in a compatibility graph---is \emph{branch and price}~\cite{Abraham07:Clearing}.  Branch and price~\cite{Barnhart98:Branch-and-Price} is a technique where only a subset of the columns (cycles and chains) are kept in the model, and promising columns are incrementally generated and added to the reduced model until optimality can be proven.  Such promising columns are found by solving the \emph{pricing problem}, which searches for at least one positive price column (i.e., variable) to add to the model---or shows that none exist. Once no more positive price variables exist, optimality has been proven for that node in the branch-and-bound search tree, and the search can proceed further in the tree.

Quickly solving the pricing problem at nodes in the search tree is important for overall runtime.  Recently, it was shown that determining whether a positive price \emph{cycle} exists can be solved in polynomial time~\cite{Glorie14:Kidney,Plaut16:Fast}. Both Glorie \emph{et al.}~\cite{Glorie14:Kidney} and Plaut \emph{et al.}~\cite{Plaut16:Fast} also use a variant of their cycle-pricing algorithms for chains.  In this paper, we show that not only are those latter algorithms incorrect, but the underlying problem---finding a positive price \emph{chain}---is, in fact, NP-complete.

\section{The pricing problem in kidney exchange}%
\label{sec:pricing}
We now formally define the pricing problem in the context of kidney exchange, but it applies to other barter markets as well.  The pricing problem is to find at least one positive price cycle or chain, or show that none exist.  The price of a cycle or chain $c$ is $\sum_{(u, v) \in c} w_{(u,v)} - \sum_{v\in c}\delta_v$, where $w_{(u,v)}$ is the weight of edge $(u,v)$, and $\delta_v$ is the dual value of vertex $v$ in the linear program  relaxation. Glorie \emph{et al.}~\cite{Glorie14:Kidney} and Plaut \emph{et al.}~\cite{Plaut16:Fast} show how determining whether a positive price cycle in the compatibility graph $G = (V,E)$ exists is equivalent to finding a negative weight cycle in a \emph{reduced graph} $G' = (V,E')$, where each edge $e' = (u,v) \in E'$ exists if and only if $(u,v) \in E$, and $e'$ has \emph{reduced weight} $r_{(u,v)} = \delta_v - w_{(u,v)}$.

A similar equivalence holds for chains. We must be careful, however, since the number of vertices in a chain exceeds the number of edges by $1$. We now define $r_{(u,v)}$ as follows:
\[r_{(u,v)} =  \begin{cases}
      \delta_v - w_{(u,v)} & u\in P\\
      \delta_u + \delta_v - w_{(u,v)} & u\in A
   \end{cases}
\]

Since an outgoing edge from an altruist will only ever be used in a chain, this ensures that a chain has positive price in $G$ if and only if it has negative weight in the reduced graph $G'$.

\section{Counterexample to two prior algorithms}
In this section, we provide counterexamples to the pricing algorithms of both Glorie \emph{et al.}~\cite{Glorie14:Kidney} and Plaut \emph{et al.}~\cite{Plaut16:Fast}.  Both previous algorithms use Bellman-Ford-style search in the reduced graph, initiated from each altruist as the source, to find negative-weight chains. Ideally, we would like to find the \emph{shortest} paths using each vertex at most once, but this is NP-hard in the presence of negative cycles~\cite{Plaut16:Fast}. However, we need not find the shortest paths beginning at each altruist: we only need to determine whether there exists any negative path starting at any altruist.

In the presence of negative cycles, traditional Bellman-Ford may generate paths with internal loops, which are invalid in our context. Plaut \emph{et al.}~\cite{Plaut16:Fast} handle this by preventing Bellman-Ford from looping during execution. As a result, the generated paths may not be the shortest, and a given negative chain may not be found.

For the version of the algorithm for cycles, Plaut \emph{et al.}~\cite{Plaut16:Fast} show that although there may be negative cycles that are not found, at least one negative cycle will be found, if any exist. Their proof of the version of the algorithm for chains is incorrect in general, however, as it implicitly assumes that the chain length cap and cycle length cap are equal.

Plaut \emph{et al.}~\cite{Plaut16:Fast} gave a counterexample to the algorithm of Glorie \emph{et al.}~\cite{Glorie14:Kidney}. Figure~\ref{fig:counterexample} gives a counterexample to the algorithm of Plaut \emph{et al.}~\cite{Plaut16:Fast}; this is also a counterexample to the original algorithm due to Glorie \emph{et al.}~\cite{Glorie14:Kidney}.

\begin{figure}[ht!bp]
\centering
\resizebox{!}{1.5 in}{ 
\begin{tikzpicture}[->,>=stealth',shorten >=1pt,auto,node distance=3cm,
  thick,main node/.style={circle,fill=blue!20,minimum size=15mm,draw,font=\sffamily\Large\bfseries}]

  \node[main node] (1) {$p_1$};
  \node[main node, fill=white] (0) [below left of=1] {$a$};
  \node[main node] (2) [below right of=1] {$p_2$};
  \node[main node] (3) [right of=2] {$p_3$};
  \node [main node] (4) [right of=3] {$p_4$};
  \node[main node] (5) [below right of=0] {$p_5$};

  \path[every node/.style={font=\sffamily\normalsize}]
    (0) edge [left] node[above left] {$0$} (1)
          edge [left] node[below left] {$1$} (5)
    (1) edge [left]  node[above right] {$0$} (2)
    (2) edge node[above] {$0$} (3)
    (3) edge [left] node[above] {$0$} (4)
    (4) edge [bend right] node[above] {$-2$} (1)
    (5) edge [right] node[below right] {$0$} (2)
      ;
\end{tikzpicture}
}
\caption{Example where the algorithm of Plaut \emph{et al.}~\cite{Plaut16:Fast} fails to find a negative chain for $L = 3$ and $K = 5$, although one exists.  Here, vertex $a$ is an altruist ($a \in A$) and the rest of the vertices are incompatible donor-patient pairs ($\{p_1,\ldots,p_5\} \in P$).}\label{fig:counterexample}
\end{figure}
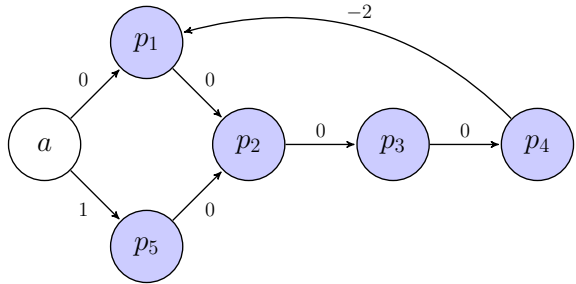

For $L = 3$ and $K=5$, there are no valid negative cycles in the reduced graph, and there is a single valid negative chain in the graph: $(a, p_5, p_2, p_3, p_4, p_1)$. Although $(p_1, p_2, p_3, p_4)$ is a cycle with negative weight, it exceeds the cycle length cap of $L=3$, and thus is invalid.

In the second iteration of the algorithm due to Plaut \emph{et al.}~\cite{Plaut16:Fast}, vertex $p_2$ would store as its most promising predecessor the path $(a, p_1, p_2)$ with weight $\pathweight{(a,p_1,p_2)}=0$, instead of $(a, p_5, p_2)$ with less promising weight $\pathweight{(a, p_5, p_2)}=1$. However, this causes the algorithm to miss the overall negative chain that would be found otherwise at iteration $5$, since it cannot reuse vertex $p_1$ (and thus cannot use the sole negative-weight edge with sink $p_1$).  Critically, even though the path $(a, p_5, p_2)$ was not promising at an earlier iteration, following it instead of the more immediately promising $(a, p_1, p_2)$ would have led to a negative-weight chain---in this case, the only negative-weight chain.  The initial algorithm due to Glorie \emph{et al.}~\cite{Glorie14:Kidney} would also incorrectly return that no negative-weight chains exist, by similar reasoning.

This shows a correctness error in both Glorie \emph{et al.}~\cite{Glorie14:Kidney} and Plaut \emph{et al.}~\cite{Plaut16:Fast}.  In the next section, we show that in general such polynomial-time approaches are hopeless: determining whether a positive price chain exists is NP-complete.

\section{Main result}

We define the \emph{negative chain problem} as follows: given a directed graph $G = (V, E)$, where $V = P\cup A$, is there a path (using each vertex at most once) of negative weight, using at most $K$ edges, and starting at some vertex $a\in A$? We call such a path a negative chain.

\begin{theorem}\label{thm:hardness}
Deciding whether a negative chain exists is NP-complete.
\end{theorem}

\begin{proof}
The negative chain problem is trivially in NP: simply sum the edge weights in a proposed path and check its sign. To show NP-hardness, we reduce from the directed Hamiltonian path problem. Given some graph $H = (V, E)$, the directed Hamiltonian path problem asks whether there exists a directed path that visits each vertex exactly once.
Let $n = |V|$.  Construct the graph $G$ as follows: set $w_e = -1$ for each $e\in E$, and add a vertex $a$ with an edge $(a, v)$ with $w_{(a,v)} = n - 2$ for each $v\in V$. Figure~\ref{fig:hamiltonian} gives an example of the construction of the graph $G$ for the proof. Let $P = V$, $A = \{a\}$, and $K = n$.

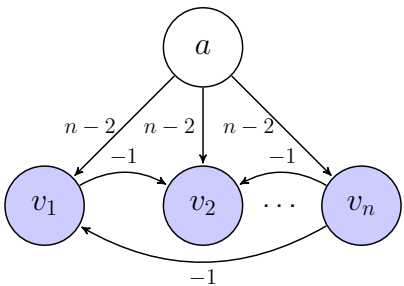
\begin{figure}[ht!bp]
\centering
\resizebox{!}{1.5 in}{ 
\begin{tikzpicture}[->,>=stealth',shorten >=1pt,auto,node distance=3cm,
  thick,main node/.style={circle,fill=blue!20,minimum size=15mm,draw,font=\sffamily\Large\bfseries}]

  \node[main node] (1) {$v_1$};
  \node[main node] (2) [right of=1] {$v_2$};
  \node[main node] (3) [right of=2] {$v_n$};
  \node [main node, fill=white] (4) [above of=2] {$a$};

  \path[every node/.style={font=\sffamily\normalsize}]
    (1) edge [bend left]  node[above] {$-1$} (2)
    (2)  -- node[auto=false]{\Large\dots} (3)
    (3) edge [bend left] node[below] {$-1$} (1)
         edge [bend right] node[above] {$-1$} (2)
    (4) edge [left] node[left] {$n - 2$} (1)
    	 edge [left] node[left] {$n - 2$} (2)
	 edge [left] node[left] {$n - 2$} (3)
      ;
\end{tikzpicture}
}
\caption{Example construction for the proof of Theorem~\ref{thm:hardness}.}\label{fig:hamiltonian}
\end{figure}

Suppose $h$ is a Hamiltonian path in $H$ starting at $v_i$. Let $c = (a, v_i) \cup h$. Since $h$ has exactly $n-1$ edges, $c$ contains $n$ edges, thereby satisfying the length constraint. Since $h$ visits each $v_i$ exactly once and never visits $a$, $c$ visits each vertex in $G$ at most once.  Finally, since $h$ has weight $1 - n$, $c$ has weight $n-2 + 1 -n = -1$. Therefore $c$ is a negative chain in $G$.

Suppose $c$ is a negative chain in $G$. Then $c$ must begin at $a$, so we can write $c = (a, v_i) \cup h$, for some $v_i \in V$ and path $h$. Let $m$ be the number of edges in $h$. Then $w_c = n - 2 - m$. Since $w_c < 0$, we have $m > n - 2$. Since $c$ can use each vertex at most once, we have $m \leq n - 1$. Therefore $m = n - 1$. Since $h$ has $n-1$ edges, $h$ visits every vertex in $V$ exactly once, making it a valid Hamiltonian path in $H$.
\end{proof}

The general pricing problem (where both cycles and chains are included) is to find a positive price (negative weight) cycle of length at most $L$ or a positive price (negative weight) chain of length at most $K$, or show that none exist. Note that solving the general pricing problem does not necessarily solve the negative chain problem. If $X$ is the set of negative chains and $Y$ is the set of negative cycles, the general pricing problem is to determine whether $X\cup Y = \emptyset$. The negative chain problem is to determine whether $X = \emptyset$: however, determining whether $X\cup Y =\emptyset$ does not necessarily determine whether $X = \emptyset$.

To show that the general pricing problem is NP-hard, we modify the above construction by expanding each edge in $H$ to a series of $L$ edges whose weights sum to \num{-1}. Then any cycle in $G$ has length at least $2L$, which violates the length constraint. Since there are no valid negative cycles in $G$, the pricing problem becomes equivalent to the negative chain problem. Therefore, the general pricing problem is also NP-hard. Since the general pricing problem is also trivially in NP, it is NP-complete.

\section{Conclusion \& implications}
We discussed branch-and-price-based approaches to the kidney exchange problem, and showed that solving the pricing problem for chains, and thereby the pricing problem for cycles and chains jointly, is NP-complete. This shows a correctness error in two leading branch-and-price-based solvers. The results apply to other barter exchanges as well, as long as they use chains (potentially with cycles as well).

Our hardness results show that a different approach for handling chains is necessary. Dickerson \emph{et al.}~\cite{Dickerson16:Position} introduce models where chains are represented by position-indexed edge variables. Since there are only a polynomial number of edge variables, they can be fully enumerated, removing the need for branch and price for chains. Cycles in these models can still be handled via branch and price, or via a different scheme.

\section*{Acknowledgements}
This material was funded by NSF grants IIS-1320620, IIS-1546752, CCF-1101668, and IIS-0964579, by the ARO under award W911NF-16-1-0061, and by a Facebook Fellowship.

\bibliographystyle{elsarticle-num}
\bibliography{dairefs}

\end{document}